\documentclass[oneside]{article}
\usepackage{amsmath}
\usepackage{amssymb}
\usepackage[margin=3cm]{geometry}
\usepackage{setspace}
\usepackage{hyperref}
\usepackage[mathscr]{euscript}
\usepackage{color}
\date{June 29, 2007}
\newtheorem{thm}{Theorem}[section]

\newtheorem{lemma}[thm]{Lemma}

\newtheorem{theorem}{Theorem}[section]

\newtheorem{corollary}[theorem]{Corollary}

\newenvironment{proof}{{\bf Proof:}}{\hfill$\square$\vskip.5cm}

\newcommand{\E}{\mathbb{E}}
\newcommand{\Z}{\mathbb{Z}}

\newcommand{\s}{\sigma}

\renewcommand{\O}{\Omega}

\title{Thermodynamic Limit for Spin Glasses.\\ Beyond the Annealed Bound}

\author{$\text{\rm Pierluigi Contucci}^1$ and 
$\text{\rm Shannon Starr}^2$\\[5pt]
\normalsize $\textcolor{white}{\text{\rm P}}^1$ Departimento di Matematica\\
\normalsize Universit\`a di Bologna, 40127 Bologna, Italy\\
\normalsize {\it email:} \texttt{contucci {\it at} dm  {\it dot} unibo {\it dot} it}\\[5pt]
\normalsize
$\textcolor{white}{\text{\rm S}}^2$ Department of Mathematics, University of Rochester\\
\normalsize RC Box 270138,
Rochester, NY 14627\\  
\normalsize {\it email:} \texttt{sstarr {\it at} math {\it dot} rochester {\it dot} edu}}

\date{14 October 2008}

\begin{document}

\onehalfspacing

\markright{}

\maketitle

\begin{abstract}
Using a correlation inequality of Contucci and Lebowitz for spin glasses, we demonstrate
existence of the thermodynamic limit for short-ranged spin glasses, under weaker hypotheses
than previously available, namely without the assumption of the annealed bound.

\vspace{8pt}
\vskip .2 cm
\noindent
{\small \bf MCS numbers:} 82B44; 60K35
\end{abstract}

\section{Introduction}

With Sandro Graffi, one of the authors proved Griffiths-type inequalities for Gaussian spin glasses in \cite{ContucciGraffi}.
These inequalities were extended to all possible spin glasses by a simple argument by Joel Lebowitz and one of the authors
in \cite{ContucciLebowitz}:
\begin{theorem}[Contucci and Lebowitz 2006]
\label{thm:CL}
Consider a spin glass model with Hamiltonian
$$
-H_{\Lambda}(\sigma,\boldsymbol{J})\,
=\, \sum_{X \subseteq \Lambda} \lambda_{X} J_{X} \sigma_X\, ,
$$
where all the $J_X$ are random and satisfy $\E[J_X]=0$, $\E[J_X^2] = 1$, and all the $\lambda_X$ are nonnegative.
Here $\sigma \in \{+1,-1\}^{\Lambda} =: \Omega_{\Lambda}$ is the spin configuration and $\sigma_X = \prod_{i \in X} \sigma_i$.
Then, for $P_{\Lambda} = \E\left[\ln \sum_{\s \in \Omega_{\Lambda}} e^{-\beta H_{\Lambda}(\sigma,\boldsymbol{J})}\right]$,
\begin{equation}
\label{ineq:CL}
\frac{\partial P_{\Lambda}}{\partial \lambda_X}\, \geq\, 0\, ,
\end{equation}
for all $X\subseteq \Lambda$.
\end{theorem}
This inequality is a perfect analogue of the first Griffiths inequality. (See, for example, \cite{Ruelle}.)
\begin{theorem}[Griffiths 1967, Kelly and Sherman 1968]
\label{thm:GKS}
Consider the Hamiltonian
$$
-H_{\Lambda}(\sigma)\,
=\, \sum_{X \subseteq \Lambda} J_{X} \sigma_X\, ,
$$
where all the $J_X$ are nonnegative.
Then, for $P_{\Lambda} = \sum_{\s \in \Omega_{\Lambda}} e^{-\beta H_{\Lambda}(\s)}$,
$$
\frac{\partial P_{\Lambda}}{\partial J_X}\, \geq\, 0\, ,
$$
for all $X\subseteq \Lambda$.
\end{theorem}
In \cite{ContucciLebowitz}, Theorem \ref{thm:CL}
was used to prove that the thermodynamic limit of a general spin-glass model exists under the 
very general condition of Thermodynamic Stability,
\begin{equation}
\label{ineq:ContLebStab}
\E[e^{-\beta H_{\Lambda}(\sigma)}]\, \leq\, e^{c|\Lambda|}\, ,
\end{equation}
for some $c<\infty$ (and all $\sigma$ by symmetry).
In the case that all of the random couplings are Gaussian,
this condition reduces to the condition of a stable potential from \cite{ContucciGraffi},
\begin{equation}
\label{ineq:CG}
\sup_{\Lambda} \frac{1}{|\Lambda|} \sum_{X \subset \Lambda} \lambda_X^2 \E[J_X^2]\,
<\, \infty\, .
\end{equation}
When the 
couplings are such that exponential moments are finite, the condition (\ref{ineq:ContLebStab})
is optimal.
If the couplings have ``fat tails,'' however, then it may occur that $\E[e^{-\beta H_{\Lambda}(\sigma)}]=\infty$ which means that 
the {\em annealed} pressure does not exist.
But this does not preclude the existence of the thermodynamic limit for the 
{\em quenched} pressure, $N^{-1} P_N(\beta)$.
In this letter, we would like to consider another inequality that allow one to prove Theorem \ref{thm:CL} with weaker assumptions.

\section{Recursive Formula for the Pressure}

The inequalities of this section are motivated by similar inequalities for mean-field diluted
spin glasses which appear, for example, in \cite{StarrVermesi}.
Let us consider a general Ising Hamiltonian.
Suppose $\Lambda \subset \Z^d$ is a finite set, and suppose $N(\Lambda)$ is some integer
and $X_1,\dots,X_{N(\Lambda)}$ are subsets of $\Lambda$,
and $J_1,\dots,J_{N(\Lambda)}$ are reals.
Let us denote $\boldsymbol{X} = (X_1,\dots,X_{N(\Lambda)})$ and 
$\boldsymbol{J} = (J_1,\dots,J_{N(\Lambda)})$.
Then we define
$$
-H_{\Lambda}(\s,\boldsymbol{X},\boldsymbol{J})\, =\, \sum_{n=1}^{N(\Lambda)} J_{n} \s_{X_n}\, .
$$
Let us define the partition function and the pressure density as
$$
Z_{\Lambda}(\beta,\boldsymbol{X},\boldsymbol{J})\, 
=\, \sum_{\s \in \Omega_{\Lambda}} e^{- \beta H_{\Lambda}(\s,\boldsymbol{X},\boldsymbol{J})}\quad
\text{and}\quad
p_{\Lambda}(\beta,\boldsymbol{X},\boldsymbol{J})\, 
=\, \frac{1}{|\Lambda|}\, \ln Z_{\Lambda}(\beta,\boldsymbol{X},\boldsymbol{J})\, .
$$
Also, let us define the Boltzmann-Gibbs measure
$$
\langle f(\sigma) \rangle_{\Lambda,\beta,\boldsymbol{X},\boldsymbol{J}}\,
=\, \sum_{\sigma \in \Omega_{\Lambda}} f(\sigma)\, \frac{e^{-\beta H_{\Lambda}(\sigma,\boldsymbol{X},\boldsymbol{J})}}
{Z_{\Lambda}(\beta,\boldsymbol{X},\boldsymbol{J})}\, .
$$
For each $n\leq N(\Lambda)$, let us denote $\boldsymbol{X}_{[n]} = (X_1,\dots,X_n)$ and $\boldsymbol{J}_{[n]} = (J_1,\dots,J_n)$.

The main inequality that we need in order to use Theorem \ref{thm:CL} is the following one-sided bound.

\begin{lemma}
\label{lem:bound}
$$
p_{\Lambda}(\beta,\boldsymbol{X},\boldsymbol{J}) - p_{\Lambda}(\beta,\boldsymbol{X}_{[n]},\boldsymbol{J}_{[n]})\,
\leq\, \frac{1}{|\Lambda|} \sum_{k=n+1}^{N(\Lambda)} 
\left[\ln \cosh(\beta J_k)  + \tanh(\beta J_k) \langle \sigma_{X_k}\rangle_{\Lambda,\beta,\boldsymbol{X}_{[k-1]},\boldsymbol{J}_{[k-1]}}\right]\, .
$$
\end{lemma}

\begin{proof}
The lemma follows by iterating
\begin{equation}
\label{ineq:basic}
p_{\Lambda}(\beta,\boldsymbol{X}_{[n+1]},\boldsymbol{J}_{[n+1]}) 
-  p_{\Lambda}(\beta,\boldsymbol{X}_{[n]},\boldsymbol{J}_{[n]})\,
\leq\, \frac{\ln \cosh(\beta J_{n+1})  + \tanh(\beta J_{n+1}) \langle \sigma_{X_n}\rangle_{\Lambda,\beta,\boldsymbol{X}_{[n]},\boldsymbol{J}_{[n]}}}
{|\Lambda|}\, .
\end{equation}
In order to prove this, note
$$
\frac{Z_{\Lambda}(\beta,\boldsymbol{X}_{[n+1]},\boldsymbol{J}_{[n+1]})}
{Z_{\Lambda}(\beta,\boldsymbol{X}_{[n]},\boldsymbol{J}_{[n]})}\,
=\, \left\langle e^{\beta J_{n+1} \sigma_{X_{n+1}}}\right\rangle_{\Lambda,\beta,\boldsymbol{X}_{[n]},\boldsymbol{J}_{[n]}}\, .
$$
But since $\sigma_{X_{n+1}}$ is either $+1$ or $-1$, we have
$e^{\beta J_{n+1} \sigma_{X_{n+1}}} = \cosh(\beta J_{n+1}) [1 + \tanh(\beta J_{n+1}) \sigma_{X_{n+1}}]$.
So
$$
\frac{Z_{\Lambda}(\beta,\boldsymbol{X}_{[n+1]},\boldsymbol{J}_{[n+1]})}
{Z_{\Lambda}(\beta,\boldsymbol{X}_{[n]},\boldsymbol{J}_{[n]})}\,
=\, \cosh(\beta J_{n+1})\left[1 + \tanh(\beta J_{n+1}) 
\left\langle \sigma_{X_{n+1}}\right\rangle_{\Lambda,\beta,\boldsymbol{X}_{[n]},\boldsymbol{J}_{[n]}}\right]\, .
$$
But since $1+x \leq e^x$, it follows that
$$
\frac{Z_{\Lambda}(\beta,\boldsymbol{X}_{[n+1]},\boldsymbol{J}_{[n+1]})}
{Z_{\Lambda}(\beta,\boldsymbol{X}_{[n]},\boldsymbol{J}_{[n]})}\,
\leq\, \cosh(\beta J_{n+1}) e^{\tanh(\beta J_{n+1}) 
\left\langle \sigma_{X_{n+1}}\right\rangle_{\Lambda,\beta,\boldsymbol{X}_{[n]},\boldsymbol{J}_{[n]}}}\, .
$$
Taking logarithms yields (\ref{ineq:basic}).
\end{proof}

\begin{corollary}
\label{cor:pressbound}
If $J_1,\dots,J_{N(\Lambda)}$ are non-random, then
\begin{equation}
\label{ineq:nonrandom}
p_{\Lambda}(\beta,\boldsymbol{X},\boldsymbol{J})\, 
\leq\,
p_{\Lambda}\left(\beta,\boldsymbol{X}_{[n]},\boldsymbol{J}_{[n]}\right)\, 
+ \frac{1}{|\Lambda|}\, \sum_{k=n+1}^{N(\Lambda)} \left[\ln \cosh(\beta J_k) + |\tanh(\beta J_k)|\right]\, .
\end{equation}
If $J_1,\dots,J_{N(\Lambda)}$ are random and dependent, then
\begin{equation}
\label{ineq:Dependent}
\E\left[p_{\Lambda}(\beta,\boldsymbol{X},\boldsymbol{J})\right]\, 
\leq\, \E\left[p_{\Lambda}\left(\beta,\boldsymbol{X}_{[n]},\boldsymbol{J}_{[n]}\right)\right]\, 
+ \frac{1}{|\Lambda|}\, \sum_{k=n+1}^{N(\Lambda)} 
\left(\E\left[\ln \cosh(\beta J_k)\right] + \E\left[\left|\tanh(\beta J_k)\right|\right]\right)\, .
\end{equation}
If $J_1,\dots,J_{N(\Lambda)}$ are independent, then
\begin{equation}
\label{ineq:Independent}
\E\left[p_{\Lambda}(\beta,\boldsymbol{X},\boldsymbol{J})\right]\, 
\leq\, \E\left[p_{\Lambda}\left(\beta,\boldsymbol{X}_{[n]},\boldsymbol{J}_{[n]}\right)\right]\, 
+ \frac{1}{|\Lambda|}\, \sum_{k=n+1}^{N(\Lambda)} 
\left(\E\left[\ln \cosh(\beta J_k)\right] + \left|\E\left[\tanh(\beta J_k)\right]\right|\right)\, .
\end{equation}
\end{corollary}

\begin{proof}
From the lemma, we know
\begin{equation}
\label{ineq:prelim}
p_{\Lambda}(\beta,\boldsymbol{X},\boldsymbol{J})\, \leq\, 
p_{\Lambda}\left(\beta,\boldsymbol{X}_{[n]},\boldsymbol{J}_{[n]}\right)\, 
+ \frac{1}{|\Lambda|}\, \sum_{k=n+1}^{N(\Lambda)} \left[\ln \cosh(\beta J_k) + \tanh(\beta J_k) 
\langle \sigma_{X_k}\rangle_{\Lambda,\beta,\boldsymbol{X}_{[k-1]},\boldsymbol{J}_{[k-1]}}\right]\, .
\end{equation}
But $|\langle \sigma_{X_k}\rangle_{\Lambda,\beta,\boldsymbol{X}_{[k-1]},\boldsymbol{J}_{[k-1]}}|\leq 1$, and this leads to (\ref{ineq:nonrandom}).
Taking expectations of that leads to (\ref{ineq:Dependent}).
If the $J_k$'s are all independent, then taking expectations of (\ref{ineq:prelim}), we obtain
\begin{align*}
\E\left[p_{\Lambda}(\beta,\boldsymbol{X},\boldsymbol{J})\right]\, 
&\leq\, \E\left[p_{\Lambda}\left(\beta,\boldsymbol{X}_{[n]},\boldsymbol{J}_{[n]}\right)\right]\\ 
&\quad + \frac{1}{|\Lambda|}\, \sum_{k=n+1}^{N(\Lambda)} 
\left(\E\left[\ln \cosh(\beta J_k)\right] + \E\left[\tanh(\beta J_k)\right] 
\E\left[\langle \sigma_{X_k}\rangle_{\Lambda,\beta,\boldsymbol{X}_{[k-1]},\boldsymbol{J}_{[k-1]}}\right]\right)\, .
\end{align*}
But $\left|\E\left[\langle \sigma_{X_k}\rangle_{\Lambda,\beta,\boldsymbol{X}_{[k-1]},\boldsymbol{J}_{[k-1]}}\right]\right|\leq 1$.
So this leads to (\ref{ineq:Independent}).
\end{proof}

\section{Application to Ferromagnets}

Before considering spin-glasses, we mention that the lemma, combined with the first Griffiths inequality,
Theorem \ref{thm:GKS}, implies the existence of the pressure for a broad range of ferromagnets.
Suppose for each finite $X \subset \Z^d$ there is a nonnegative coupling $J_X$ such that $J_{\tau_j(X)}=J_X$
where $\tau_j(X) = \{i+j\, :\, i \in X\}$.
In other words, the couplings are translation invariant.
We also suppose $J_{\emptyset}=0$ for simplicity.
Define
$$
H_{\Lambda}(\s)\, =\, \sum_{X \subseteq \Lambda} J_X \s_X\, ,
$$
for all $\s \in \O_{\Lambda}$.
Let $p_N(\beta) = p_{[1,N]^d}(\beta)$.
Suppose that $N_1\geq 1$ is fixed and $N =  N_1 m + r$ where $m\leq 0$ and $r\geq 0$.
Then one has following inclusion
$$
[1,N]^d\, \supseteq \bigsqcup_{j \in [0,m-1]^d} \tau_{N_1 j}([1,N_1]^d)\, ,
$$
where $N_1 j = (N_1 j_1,\dots,N_1 j_d)$ for $j=(j_1,\dots,j_d)$.
Let us define an inequality $H_{\Lambda} \succcurlyeq H_{\Lambda}'$ if the couplings
in $H_{\Lambda}$ are all greater than or equal to the couplings in $H_{\Lambda}'$.
Therefore, we see that
$$
H_{[1,N]^d}\, \succcurlyeq\, \sum_{j \in [0,m-1]^d} H_{\tau_{N_1 j}([1,N_1]^d)}\, ,
$$
because in the latter we simply set $J_X$ to $0$ for any $X$ that does not fit
entirely in one $\tau_{N_1 j}([1,N_1]^d)$.
Then the first Griffiths inequality, combined with an interpolation argument,
implies that the thermodynamic potential of the full Hamiltonian $H_{\Lambda}$ dominates the 
thermodynamic potential of the right hand side, which specifically means
$$
p_N(\beta)\, \geq\, \frac{(mN_1)^d}{N^d}\, p_{N_1}(\beta) + \left(1 - \frac{(m N_1)^d}{N^d}\right) \ln(2)\, .
$$
We can choose $m$ so that $N - mN_1 \leq N_1-1$. Therefore, $\lim_{N \to \infty} (mN_1/N) = 1$.
So, we deduce
$$
\liminf_{N \to \infty} p_N(\beta)\, \geq\, p_{N_1}(\beta)\, .
$$
Since this was true for all $N_1$, this implies
$$
\liminf_{N \to \infty} p_N(\beta)\, \geq\, \sup_{N_1\geq 1} p_{N_1}(\beta)\, \geq\, \limsup_{N \to \infty} p_N(\beta)\, .
$$
This means that $\lim_{N \to \infty} p_N(\beta)$ exists and equals $\sup_{N\geq 1} p_N(\beta)$.
Of course, it is possible that the supremum, and hence the limit, may equal $\infty$.
That is precisely where we use the bounds from the last section.

\begin{corollary}
\label{cor:Ferro}
Define
$$
\|J\|\, :=\, 
\sum_{\substack{X \subset \Z^d \\ X \ni 0}} \frac{J_X}{|X|}\, .
$$
If $\|J\|<\infty$ then $p(\beta) = \lim_{N \to \infty} p_N(\beta)$ exists as a finite number and $p(\beta)\, \leq\, \ln(2) + 2 \beta \|J\|$.
\end{corollary}

\begin{proof}
Enumerate the subsets $X\subseteq \Lambda$ any way, in order to express $H_{\Lambda}(\s)$ as 
$H_{\Lambda}(\s,\boldsymbol{X},\boldsymbol{J})$, as in the last section.
Note that $H_{\Lambda}(\s,\boldsymbol{X}_{[0]},\boldsymbol{J}_{[0]}) = 0$ 
so that 
$p_{\Lambda}(\beta,\boldsymbol{X}_{[0]},\boldsymbol{J}_{[0]})=\ln(2)$.
Then taking $n=0$ in (\ref{ineq:nonrandom}) and noting that
$\ln \cosh(x) \leq |x|$ and $|\tanh(x)| \leq |x|$, we obtain
$$
p_{\Lambda}(\beta)\, \leq\, \ln(2) + \frac{2\beta}{|\Lambda|} \sum_{X \subseteq \Lambda} J_X\, .
$$
But, of course, using translation invariance and the definition of $\|J\|$,
$$
\frac{1}{|\Lambda|} \sum_{X \subseteq \Lambda} J_X\, =\, \frac{1}{|\Lambda|} \sum_{X \subseteq \Lambda} \sum_{i \in X} \frac{J_X}{|X|}\,
=\, \frac{1}{|\Lambda|} \sum_{i \in \Lambda} \sum_{\substack{X \subset \Lambda \\ X \ni i}} \frac{J_X}{|X|}\,
\leq\, \|J\|\, .
$$
So each $\Lambda$ has the bound $p_{\Lambda}(\beta) \leq \ln(2) + 2 \beta \|J\|$, and hence $\sup_{N\geq 1} p_{[1,N]^d}(\beta)$ also satisfies
the bound.
\end{proof}
For the ferromagnet, the result obtained in the corollary is well-known.
See, for example, \cite{Ruelle}.
But for spin glasses, the same line of reasoning leads to new results.

\section{Application to Spin Glasses}

\begin{corollary}
\label{cor:SG2ndMoment}
Suppose that there are independent, random couplings $J_X$ for each finite subset $X \subset \Z^d$,
which are centered and such that the distributions are translation invariant.
(For simplicity, suppose $J_{\emptyset} \equiv 0$, again.)
Define
$$
\|J\|_2^2\, =\, \sum_{\substack{X \subset Z^d \\ X \ni 0}} \frac{\E[J_X^2]}{|X|}\, .
$$
Assuming $\|J\|_2^2<\infty$, we have $p(\beta) = \lim_{N \to \infty} p_{[1,N]^d}(\beta)$ exists as a finite number
and satisfies the bound $p(\beta) \leq \ln(2) + \frac{3 \beta^2}{2} \|J\|_2^2$, where $p_{\Lambda}(\beta)$ is the quenched pressure
$$
p_{\Lambda}(\beta)\, =\, \E[p_{\Lambda}(\beta,\boldsymbol{J})]\, .
$$
\end{corollary}

This corollary is comparable to results of Khanin and Sinai \cite{KhaninSinai} and van Enter and van Hemmen \cite{vanEntervanHemmen},
except that we do not attempt to prove convergence in the van Hove sense, settling instead for convergence in the Fisher
sense (see, for example, \cite{vanEntervanHemmen, vanEnterFernandezSokal} for the difference), but also, we 
do not make any conditions on finite moments
of the random couplings beyond existence of the variance.

\begin{proof}
Using the CL inequality (\ref{ineq:CL}), we conclude that $\lim_{N \to \infty} p_N(\beta) = \sup_{N\geq 1} p_N(\beta)$,
where $p_N(\beta)$ is the quenched pressure $p_{[1,N]^d}(\beta)$, using the same argument as in the last section.
All that remains is to obtain bounds.
From (\ref{ineq:Independent}), we know
\begin{equation}
\label{ineq:SG2ndFirst}
p_{\Lambda}(\beta)\, =\, \E[p_{\Lambda}(\beta,\boldsymbol{J})]\, 
\leq\, \ln(2) + \frac{\beta}{|\Lambda|}\, \sum_{X \subseteq \Lambda} \left(\E[\ln\cosh(\beta J_X)] + |\E[\tanh(\beta J_X)]|\right)\, .
\end{equation}
Now we know $\E[\ln\cosh(\beta J_X)] \leq \beta^2 \E[J_X^2]/2$ because $\ln \cosh(x) \leq x^2/2$.
Also, by assumption, we know $\E[J_X]=0$.
So 
$$
|\E[\tanh(\beta J_X)]|\, =\, |\E[\tanh(\beta J_X) - \beta J_X]|\, \leq\, \E\left[|\tanh(\beta J_X) - \beta J_X|\right]\, .
$$
But $|x-\tanh(x)| = \int_0^{|x|} \tanh^2(y)\, dy \leq |x| \tanh^2(x) \leq \min(|x|,|x|^3)$.
It is easy to see that $\min(|x|,|x|^3) \leq x^2$.
So we obtain, $|\E[\tanh(\beta J_X)]| \leq \beta^2 \E[J_X^2]$.
Combining these bounds with (\ref{ineq:SG2ndFirst}) leads to the desired bound.
\end{proof}

The main improvement over previous results by Khanin and Sinai, and van Hemmen and van Enter,
is that we weakened the hypotheses on the moments of the random couplings $J_X$.
In fact, the condition in the corollary is just the specialization of the
thermodynamic stability condition (\ref{ineq:CG}) from \cite{ContucciGraffi}, specialized to translation-invariant distributions
for the couplings.
Therefore, it is optimal.
However, one can imagine a situation with even fatter tails, so that even the variance does not exist.
In that case, we can apply the following corollary:

\begin{corollary}
\label{cor:SG1stMoment}
Suppose that there are independent, random couplings $J_X$ for each finite subset $X \subset \Z^d$,
which are centered and such that the distributions are translation invariant (and $J_{\emptyset} \equiv 0$).
Define
$$
\|J\|_1\, =\, \sum_{\substack{X \subset Z^d \\ X \ni 0}} \frac{\E[|J_X|]}{|X|}\, .
$$
Assuming $\|J\|_1<\infty$, the thermodynamic limit of the quenched pressure exists,
$p(\beta) = \lim_{N \to \infty} p_{[1,N]^d}(\beta)$,
and satisfies the bound $p(\beta) \leq \ln(2) + 2 \beta \|J\|_1$. 
\end{corollary}

\begin{proof}
In this case the upper bound on the quenched pressure is easy.
We already know that 
$$
p_{\Lambda}(\beta,\boldsymbol{J})\, \leq\, \ln(2) + \frac{2\beta}{|\Lambda|}\, \sum_{X \subseteq \Lambda} |J_X|\, ,
$$
as in the proof of Corollary \ref{cor:Ferro}.
Taking expectations leads to $p_{\Lambda}(\beta) := \E[p_{\Lambda}(\beta,\boldsymbol{J})] \leq \ln(2) + 2 \beta \|J\|_1$.
However, in this case, we need to check that (\ref{ineq:CL}) still applies.
The situation for the CL inequality was that $\E[J_X^2]=1<\infty$ for all $X$ (and then the random variables
were scaled by multipliers $\lambda_X$).
It is intuitively obvious that an integrated version of the inequality still holds for the case where the first moment is finite, but not the second.
So let us quickly prove it.
Define the centered truncation $J_X^{(1)} = J_X \cdot \chi_{[-R,R]}(J_X) - \E[J_X \cdot \chi_{[-R,R]}(J_X)]$, where $R<\infty$ is arbitrary.
Let $J_X^{(2)} = J_X - J_X^{(1)}$.
Clearly $J_X^{(1)}$ does have a second moment.
Therefore,
we know that (\ref{ineq:CL}) is true if we replace all $J_X$'s by $J_X^{(1)}$'s.
In particular, this means we have the necessary type of super-additivity as long as we replace the $J_X$'s by $J_X^{(1)}$'s.
All that remains is to check that if we take $R \to \infty$, we recover the original pressure.
But, using (\ref{ineq:Dependent}), we see that
\begin{equation}
p_{\Lambda}(\beta) - p_{\Lambda}^{(1)}(\beta)\, 
\leq\, \frac{1}{|\Lambda|}\, \sum_{X \subseteq \Lambda} 
\left(\E\left[\ln \cosh(\beta J_X^{(2)})\right] + \E\left[\left|\tanh(\beta J_X^{(2)})\right|\right]\right)\, 
\leq\, \frac{2\beta}{|\Lambda|}\, \sum_{X \subseteq \Lambda} \E[|J_X^{(2)}|]\, .
\end{equation}
Note that $J_X^{(2)}$ is dependent on $J_X^{(1)}$, but (\ref{ineq:Dependent}) applies in this case.
All the $|J_X^{(2)}|$'s are dominated by $|J_X| + \E[|J_X|]$.
(The $\E[|J_X|]$ comes from the shift to $J_X^{(1)}$.)
Also, clearly as $R \to \infty$, we have $J_X^{(2)} = J_X \cdot [1 - \chi_{[-R,R]}(J_X)]  - \E[J_X \cdot (1 - \chi_{[-R,R]}(J_X))]$
converging to $0$, almost surely.
So by the Dominated Convergence Theorem, it is true that $\lim_{R \to \infty} p_{\Lambda}^{(1)}(\beta) = p_{\Lambda}(\beta)$.
Therefore, one recovers the Contucci-Lebowitz super-additivity in the limit.
\end{proof}

We could easily combine the two types of results. 
\begin{corollary}
Suppose $J_X$ and $J_X'$ are random, centered couplings (with $J_{\emptyset}\equiv J_{\emptyset}' \equiv 0$)
which are all independent and such that $\|J\|_2^2 < \infty$ and $\|J'\|_1 < \infty$.
Then taking
$$
H_{\Lambda}(\s)\, =\, \sum_{X \subseteq \Lambda} (J_X + J_X') \s_X\, ,
$$
the thermodynamic limit of the quenched pressure exists and satisfies $p(\beta) \leq \ln(2) + \frac{3\beta^2}{2}\, \|J\|_2^2 + 2 \beta \|J'\|_1$.
\end{corollary}

\begin{proof}
This is the setting of Corollary \ref{cor:pressbound}.
Namely, define $N(\Lambda) = 2 \cdot 2^{|\Lambda|}$ and let $X_1,\dots,X_{2^{|\Lambda|}}$
and $X_{2^{|\Lambda|}+1},\dots,X_{2\cdot 2^{\|\Lambda|}}$ each enumerate the subsets of $\Lambda$, independently.
Similarly, let $J_n = J_{X_n}$ for $n\leq 2^{|\Lambda|}$, $J_{n} = J'_{X_{n}}$ for $n> 2^{|\Lambda|}$.
Then the Hamiltonian is defined as $H_{\Lambda}(\sigma,\boldsymbol{X},\boldsymbol{J})$.
The bounds from before then imply the result.
\end{proof}

The result of the previous corollary reproduces a main result from the paper \cite{Zegarlinski} by Zegarlinski.
(We thank A.C.D.~van Enter for bringing this to our attention.)
However, our proof uses the  Griffiths-type inequality (\ref{ineq:CL}) for spin glasses,
which seems to give a simpler, more modern approach.
We can also easily interpolate the results to obtain the following.
\begin{corollary}
Define $\|J\|_p^p = \sum_{X \ni 0} |X|^{-1} \E[|J_X|^p]$. As long as $\|J\|_p < \infty$ for some $1\leq p \leq 2$, then the thermodynamic limit
of the pressure exists in the Fisher sense.
\end{corollary}

\begin{proof}
We may bound both $|\tanh(x) - x|$ and $\ln \cosh(x)$ by some constants times $\min(|x|,|x|^2)$.
Therefore, for any $1\leq p\leq 2$, we may bound these functions by some constant time $|x|^p$, and this
suffices to derive an upper bound on $p_{\Lambda}(\beta)$ in terms of $\|J\|_p$, which is uniform in $\Lambda$.
\end{proof}

\section*{Acknowledgments}
S.S.~was supported in part by a U.S.\ National Science Foundation
grant, DMS-0706927. P.C.~ was partially supported by the European Grant Cultaptation
and the Unibo Strategic Research plan. 
Part of this research was carried out while P.C.~and S.S.~visited the Weierstrass Institute for Applied Analysis and Stochastics (WIAS)
in Berlin, and while S.S.~visited the Erwin Schr\"odinger International Institute for Mathematical Physics (ESI)
in Vienna.
We thank them for their warm hospitality and thank Alessandra Bianchi and Anton Bovier for making the visit to WIAS possible.
We also thank Aernout van Enter for his helpful comments on an earlier preprint draft of this paper.

\end{document}